\documentclass[a4paper]{article}
\usepackage{graphics}
\usepackage{epsfig}
\usepackage{amsmath}
\usepackage{amssymb}
\usepackage{fancybox}
\usepackage{bm}
\usepackage{amsthm}
\usepackage{mathrsfs}

\def\T{\mathscr{T}}
\def\f{\mathscr{F}}

\def\I{\mathscr{I}}

\def\M{{\cal M}}
\def\MM{\mathscr{M}}

\def\K{\mathbb{K}}

\def\R{\mathbb{R}}
\def\C{\mathbb{C}}
\def\Z{\mathbb{Z}}
\def\N{\mathbb{N}}
\def\LL{{\cal L}}

\def\HH{{\cal H}}

\newtheorem{principle}{Principle}
\newtheorem{remark}{Remark}

\newtheorem{theorem}{\bf Theorem}[section]

\newtheorem{definition}{\bf Definition}[section]

\newtheorem{example}{\bf Example}[section]
\begin{document}
\title{
Thermodynamics for Trajectories of a Mass Point
}
\author{Yoshimasa Kurihara\footnote{yoshimasa.kurihara@kek.jp}\\
{\footnotesize\it The High Energy Accelerator Organization (KEK), 
Tsukuba, Ibaraki 305-0801, Japan}\\ 
$~$\\
Khiem Hong Phan and Nhi My Uyen Quach\\
{\footnotesize\it The Graduate University for Advanced Studies, Tsukuba, Ibaraki 305-0801, Japan}}
\date{}
\maketitle
\begin{abstract}
On the basis of the theory of thermodynamics, a new formalism of classical nonrelativistic mechanics of a mass point is proposed. The particle trajectories of a general dynamical system defined on a $(1+n)$-dimensional smooth manifold are geometrically treated as dynamical variables. The statistical mechanics of particle trajectories are constructed in a classical manner. Thermodynamic variables are introduced through a partition function based on a canonical ensemble of trajectories. Within this theoretical framework, classical mechanics can be interpreted as an equilibrium state of statistical mechanics. The relationship between classical and quantum mechanics is discussed from the viewpoint of statistical mechanics. The maximum-entropy principle is shown to provide a unified view of both classical and quantum mechanics.
\end{abstract}
\section{Introduction}
Quantum mechanics is considered to be the most basic theory of nature. All phenomena, including gravitational interactions, have an underlying quantum-mechanical interpretation. Quantum mechanics describes the microscopic behavior of particles under fundamental forces and has been adopted in numerous applications. However, our understanding of quantum mechanics remains incomplete. One of the most characteristic and mysterious aspects of quantum mechanics is that particle properties are described by probability amplitudes. The probabilistic aspects of quantum mechanics are inherent characteristics and are not due to lack of detailed information such as partition functions in statistical mechanics. Therefore, understanding why and how the probabilistic nature of quantum mechanics emerges from a primary principle is of critical importance. To pursue this purpose, we propose herein to use a thermodynamic theory.
\\
Let us recall the relationship between thermodynamics and statistical mechanics. Thermodynamics is a field of physics that discusses the relationship between the macroscopic physical quantities such as temperature, pressure, volume, energy, entropy, and heat and/or work from outside of the system. Thermodynamics was established before the microscopic details were clarified. Later, statistical mechanics was constructed on the basis of the microscopic details of classical mechanics using thermodynamics as a guiding principle. However, statistical mechanics based on classical mechanics failed to explain, for instance, the entire nature of electromagnetic waves radiated from gases and metals, which provided a hint about quantum mechanics. Although the microscopic details of statistical mechanics were replaced by quantum mechanics instead of by classical mechanics, the consequences of thermal dynamics remain true, and again thermal dynamics plays the role of a guiding principle in constructing a theory. Thermodynamics and its second law are still active in the field of physics, and are discussed, for instance, by Lieb and Yngvason in their epoch-making paper\cite{Lieb19991}. Relationships between thermodynamics and quantum mechanics are also intensively discussed by Gemmer, Michel and Mehler\cite{QTD}. They gave new explanations for the emergence of thermodynamics behavior from a quantum mechanical system. Our way is something opposite to their intention, `` the emergence of quantum mechanics from the thermodynamics''.
\\
In our pursuit of the principle that underlies the probabilistic characteristics of the basic equations of motion, we suspend the notion that quantum mechanics is the most fundamental theory of nature and regard it as a phenomenological theory. In other words, we propose to construct the thermodynamics of quantum mechanics and then pursue the underlying mechanics, which must be a more fundamental theory of nature. In this study, we attempt to construct the thermodynamics of classical and quantum mechanics of a mass point. To construct the thermodynamics of classical mechanics, an appropriate definition of entropy must be introduced. To determine the entropy of the classical dynamical system of a mass point, we consider the system in geometrical terms and introduce a general dynamical system $\acute{a}~la$ Arnol'd, Vogtmann and Weistein\cite{arnol1989mathematical}. Then, the analogy between the thermodynamics of gases and the Hamiltonian formalism of classical dynamics guides us to identify an appropriate definition of entropy. An equation of motion of the system can be extracted by requiring the maximum-entropy principle (instead of the principle of least action). Finally, quantum fluctuations can be understood in analogy with thermal fluctuations. This view may bring new insights of quantum mechanics.
\\
This paper is organized as follows:
After introducing a general framework for dynamic systems in Section 1, a generalized Hamiltonian formalism is developed in Section 2. The statistical mechanics of particle trajectory in the proposed framework is developed in Section 3. Here, we demonstrate that a thermal-equilibrium state of the trajectory corresponds to the classical mechanics of a mass point. Section 4 is devoted to relating classical and quantum mechanics using the statistical mechanics analogy. 
Conclusions and further discussion are presented in Section 5.

\section{General Mechanics}
%
%
\subsection{Dynamics on a Symplectic Manifold}\label{DYonSyMan}
First, we introduce a general dynamic space on which various dynamical systems are developed. A structure given here is a common for classical dynamic systems treated in this report. Is is considered to give the minimum mathematical system  to discuss classical mechanics.
\begin{definition}\label{sympmani}{\bf General Dynamic Space}\\
A general dynamical space is a collection of sets, $\{\M=\MM\otimes\T,f,{\cal P}=\bm{\xi}\otimes\bm{\eta},(\omega,\Omega)\}$ whose elements are defined as follows:
\end{definition}
\begin{enumerate}
\item A manifold $\MM\in\R^n$ is an $n$-dimensional Euclidean space called a space manifold.
\item A manifold $\T\in\R$ is a one-dimensional smooth manifold called a time manifold. A point on $\T$ is called an order parameter or time. 
\item The direct product $\M=\T\otimes\MM$ is a space-time manifold. A position vector $\bm\xi$ on an open neighborhood $U_p$ of a point $p$ can be expressed in terms of local orthonormal bases as 
\begin{eqnarray*}
\bm\xi=(\xi^0=\tau,\xi^1,\cdots,\xi^n).
\end{eqnarray*}
A flat metric whose metric tensor $g_{\mu\mu}=(-,+,\cdots,+)$ is imposed on the space-time manifold. Then, the $\M$ becomes a Riemannian manifold with an indefinite metric. 
\item A tangent bundle of $\M$ is written as $T\M=\bigcup_{p\in\M}T_{p}\M$. A tangent vector to $\bm\xi$ is expressed as 
\begin{eqnarray*}
\partial_{\bm\xi}&=&\left(\frac{\partial}{\partial\xi^0}=\frac{\partial}{\partial\tau},\frac{\partial}{\partial\xi^1},\cdots,\frac{\partial}{\partial\xi^n}\right)\in T\M.
\end{eqnarray*}
\item In the same manner, we introduce a cotangent bundle, $T^*\M$, with cotangent vector to $\bm\xi$ expressed as
\begin{eqnarray*}
d\bm\xi&=&\left(d\xi^0=d\tau,d\xi^1,\cdots,d\xi^n\right)\in T^*\M.
\end{eqnarray*}
\item A characteristic function $f$ is a $C^\infty$-function that maps a point on $\M$ to a real number. 
The characteristic function is assumed holomorphic for a position vector $\bm\xi$ such that $\partial^2 f/\partial\xi^\mu\partial\xi^\nu=\partial^2 f/\partial\xi^\nu\partial\xi^\mu,~\mu,\nu=0,\cdots,n$.
\item A momentum vector is introduced in terms of the characteristic function as
\begin{eqnarray*}
\bm\eta&=&\left(\eta^0,\cdots,\eta^n \right),\\
&=&\partial_{\bm\xi^*}f=\left(\frac{\partial f}{\partial\xi_0},\cdots,\frac{\partial f}{\partial\xi_n}\right),
\end{eqnarray*}
where $\bm\xi^*=\{g_{\mu\nu}\xi^\mu\}$. Here we follow the Einstein convention in summing the repeated indices, summing the Greek indices from $0$ to $n$, and summing the Roman indices from $1$ to $n$, unless otherwise stated. Not all $\eta$'s and $\xi$'s are independent because the characteristic function imposes a constraint. We assume that the zeroth component of the momentum vector is a function of other components, $\eta^0(\eta^1,\cdots,\eta^n,\xi^0,\cdots,\xi^n)$.
A direct product of position and momentum vector space ${\cal P}=\bm\xi\otimes\bm\eta$ is called an extended phase space. This space is a $(2n+2)$-dimensional smooth manifold.
\item On the extended phase-space, $1$- and $2$-forms such as
\begin{eqnarray}
\omega&=&
\bm\eta\cdot d\bm\xi\\
&=&\eta_\mu d\xi^\mu\nonumber\\
&=&\eta_id\xi^i-\eta_0d\tau,\label{1form}\\
\Omega&=&d\omega\nonumber \\
&=&d\eta_\mu \wedge d\xi^\mu\nonumber \\
&=&d\eta_i\wedge d\xi^i-d\eta_0\wedge d\tau,\label{2form}
\end{eqnarray}
are defined\cite{arnol1989mathematical}, which are called characteristic $1$-form and $2$-form, respectively.
The characteristic $2$-form is assumed to be a closed form satisfying $d\Omega=0$. The even dimensional manifold that has a non-degenerate and closed $2$-form is called a symplectic manifold.
\end{enumerate}
\begin{definition}{\bf Time evolution operator and Trajectory}\\
The smooth map 
\begin{eqnarray*}
\phi_s:\M\rightarrow\M:{\bm \xi}=(\tau,\xi^1,\cdots,\xi^n)
\mapsto\phi_s{\bm \xi}={\bm \xi}'=(\tau+s,\xi'^1(\tau+s),\cdots,\xi'^n(\tau+s)),
\end{eqnarray*}
is called a time evolution operator.
The map $\phi_s$ generates a one-dimensional manifold $\bm\gamma\subset\M$, such that
\begin{eqnarray*}
\bm\gamma=\{\bm\xi(\tau)=(\tau,\xi^1,\cdots,\xi^n)(\tau)~\big|~\phi_s(\tau_1,\xi'^1(\tau_1+s),\cdots,\xi'^n(\tau_1+s), s\in[0,\tau_2-\tau_1]\}.
\end{eqnarray*}
This manifold is called a trajectory. 
\end{definition}
We assume that a tangent vector exists for a given trajectory, i.e.,
\begin{eqnarray*}
\frac{d\bm\xi(\tau)}{d\tau}\Bigl|_{\tau=t}&=&
\lim_{\delta t\to\pm 0}\frac{\bm\xi(t+\delta t)-\bm\xi(t)}{\delta t},
\end{eqnarray*} 
at any $t\in[\tau_1,\tau_2]$.
The time evolution operator maps a momentum vector as
\begin{eqnarray*}
\bm\eta'=\phi_s\bm\eta=\partial_{\bm\xi^*}f\Big|_{\bm\xi=\phi_s\bm\xi}.
\end{eqnarray*}
The time evolution operator, which introduces dynamics to the general dynamic space, must describe some physical principle. More specifically, the following is true.
\begin{principle}\label{calpri}{\bf (Cartan)}{\rm \cite{caltan, 1998kaisekirikigaku}}\\
When the integration of the characteristic $1$-form is invariant under map $\phi_s$, i.e.,
\begin{eqnarray*}
\int_l \omega({\bm\xi,\bm\eta})&=&
\int_{l(\phi_s\bm\xi,\phi_s\bm\eta)}\omega(\phi_s\bm\xi,\phi_s\bm\eta),
\end{eqnarray*}
the trajectory induced by $\phi_s$ is physically realized.
Here $l=l(\bm\xi,\bm\eta)$ is any closed circle in an extended phase space $(\cal P)$ at fixed $\tau$.
\end{principle}
\begin{theorem}{\bf (Characteristic Equations)}\\
Trajectories satisfying Principle~\ref{calpri} determine an equation of motion such that
\begin{eqnarray*}
\frac{d\tilde\xi^i}{ds}=\frac{\partial \tilde\eta_0}{\partial \eta^i},~~
\frac{d\tilde\eta_i}{ds}=-\frac{\partial\tilde\eta_0}{\partial\xi^i},~~
\frac{d\tilde\eta_0}{ds}=\frac{\partial\tilde\eta_0}{\partial s},
\end{eqnarray*}
which are known as characteristic equations. 
\end{theorem}
Here ${\bm \tilde\eta}=(\tilde\eta_0,\tilde\eta_1(s),\cdots,\tilde\eta_n(s))$ is the tangent vector along the trajectory $\tilde\xi(s)$, defined as
\begin{eqnarray*}
\tilde\eta(s)&=&\partial_{\bm\xi^*}f\Big|_{\bm\xi=\tilde\xi(s)},
\end{eqnarray*}
where $\tilde\eta^0=\partial f/\partial\tau|_{\tau=s}$ is assumed to be a function of $\bm\xi$ and $\eta_i,~(i=1,\cdots,n)$.
\begin{figure}[tb]
  \begin{center} 
    \includegraphics[height=5cm]{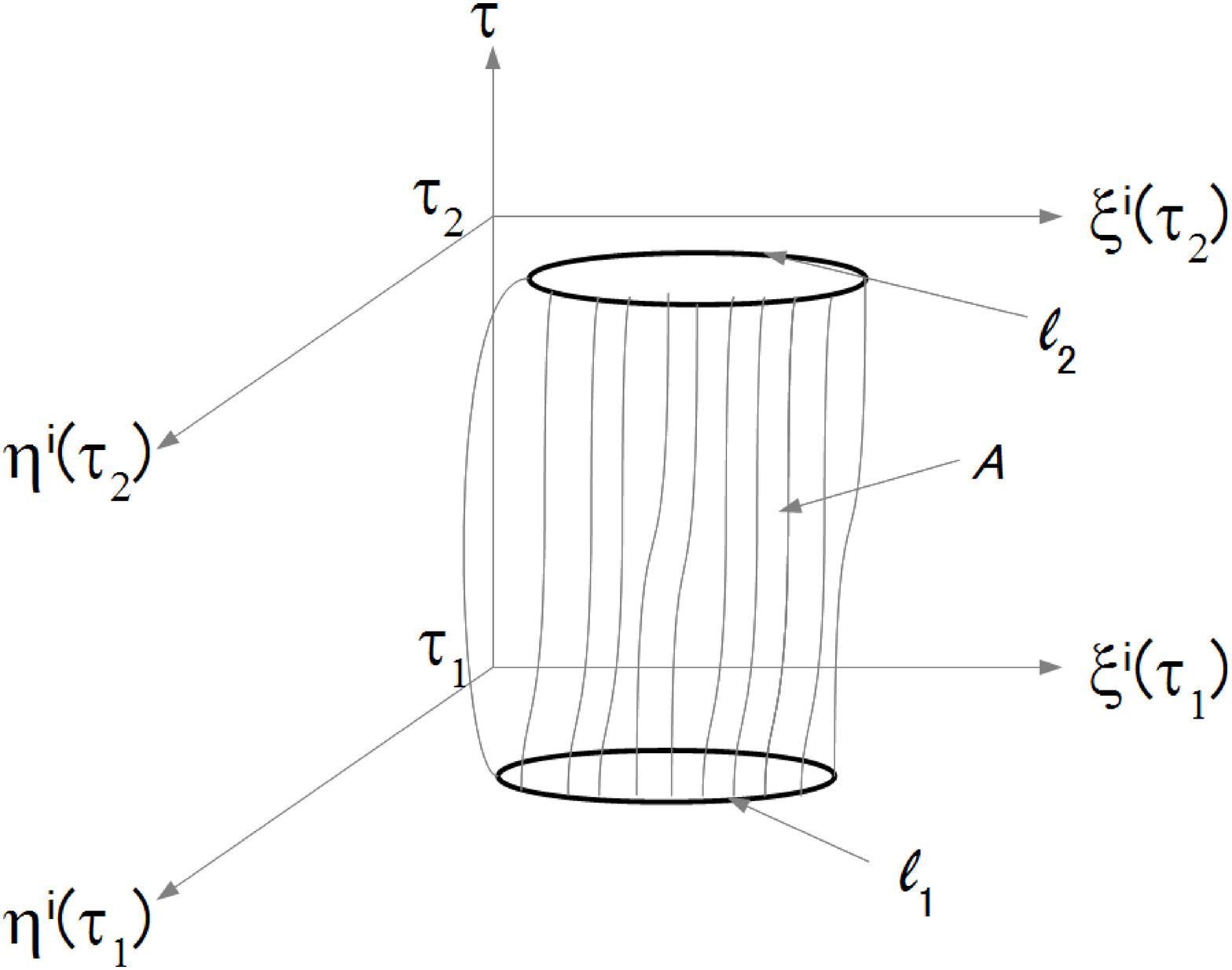}
    \caption{\footnotesize Cartan tube.} 
    \label{tube}
  \end{center}
\end{figure}
\begin{proof}
The characteristic forms along the trajectory $\tilde\xi(s)$ can be written as 
\begin{eqnarray*}
\tilde\omega&=&\tilde\eta_\mu d\tilde\xi^\mu\\
&=&-\tilde\eta_0d\tau
+\tilde\eta_id\tilde\xi^i,\\
\tilde\Omega&=&d\tilde\eta_\mu\wedge d\tilde\xi^\mu\\
&=&-d\tilde\eta_0\wedge d\tilde\xi^0+d\tilde\eta_i\wedge d\tilde\xi^i,
\end{eqnarray*}
and by definition
\begin{eqnarray}
d\tilde\xi^0&=&d\tau,\nonumber \\
d\tilde\eta_0&=&
\left[\frac{\partial\tilde\eta_0}{\partial\xi^i}d\xi^i
+\frac{\partial\tilde\eta_0}{\partial\eta_i}d\eta_i
+\frac{\partial\tilde\eta_0}{\partial \tau}d\tau
\right]_{\tau=s},
\label{f1}\\
&=&
\frac{\partial\tilde\eta_0}{\partial\xi^i}d\tilde\xi^i
+\frac{\partial\tilde\eta_0}{\partial\eta_i}d\tilde\eta_i
+\frac{\partial\tilde\eta_0}{\partial s}ds.\label{f2}
\end{eqnarray}
Here we used a property of an exterior derivative,
$
d(\omega \wedge \eta)=d\omega+(-1)^{deg~\omega}(\omega \wedge d\eta),
$
for any $\omega$ and $\eta$, and $\partial\tilde\xi_0/\partial s=1$.
Note that eq.(\ref{f2}) is a simplified expression of eq.(\ref{f1}).
Consider the {\it Cartan tube} shown in Fig.~$\ref{tube}$. Let $l_1$ be a closed circle at $\tau=\tau_1$, and $l_2=\phi_{\tau_2-\tau_1}({\bm\xi}\in l_1)$. The cylindrical surface of the Cartan tube $A$ consists of trajectories connected from $l_1$ to $l_2$. By Stokes' theorem, an integration of the characteristic $2$-form on $A$ vanishes \footnote{See for example \cite{flanders1963differential}.}, i.e.,
\begin{eqnarray}
\int_A \tilde{\Omega}&=&\int_A d\tilde{\omega},\nonumber \\
&=&\int_{\partial A} \tilde{\omega},\nonumber \\
&=&\int_{l_2} \tilde{\omega}
-\int_{l_1} \tilde{\omega},\nonumber \\
&=&0,\label{f4}
\end{eqnarray}
where the last equivalence follows from the definition of the trajectory, which maintains constant $\tilde\omega$. On the other hand, by contraction of the left-hand side of eq.(\ref{f4}) with $\phi_s=ds\partial/\partial s$, we obtain
\begin{eqnarray}
0&=&
\int_A \langle\tilde{\Omega}|\phi_s\rangle,\label{f5-1}\\
&=&\int_A\left(
-\frac{d\tilde\eta_0}{ds}d\tilde\xi^0
+\frac{d\tilde\eta_i}{ds}d\tilde\xi^i
-d\tilde\eta_0\frac{d\tilde\xi^0}{ds}
+d\tilde\eta_i\frac{d\tilde\xi_i}{ds}
\right),\nonumber \\
&=&
\int_A\left(
\frac{\partial\tilde\eta_0}{\partial s}-\frac{d\tilde\eta_0}{ds}
\right)ds
+\int_A\left(
\frac{d\tilde\eta_i}{ds}+\frac{\partial\tilde\eta_0}{\partial\xi^i}
\right)d\tilde\xi^i
+
\int_A\left(
\frac{\partial\tilde\eta_0}{\partial\eta^i}-\frac{d\tilde\xi_i}{ds}
\right)d\tilde\eta_i,\nonumber
\end{eqnarray}
where $\langle\bullet|\bullet\rangle$ is a contraction of two forms.  
The third step of this contraction uses eq.(\ref{f2}).
For the integral to identically vanish on any $A$, each term in the parentheses must be zero. 
\end{proof}
\noindent
This proof immediately leads to the following remark.
\begin{remark}{\bf (Coordinate Independent Representation)}\\
The characteristic equations can be expressed independently of the coordinates as
\begin{eqnarray*}
\langle\tilde{\Omega}|\phi_s\rangle=0.
\end{eqnarray*}
\end{remark}
\begin{proof}
The proof is evident from eq.(\ref{f5-1}), which is true for all $A$.
\end{proof}
\noindent
Hereafter, the $\tilde~$ of the trajectory is omitted for simplicity.
\begin{definition}\label{HLA}{\bf (Hamiltonian, Lagrangian and Action)}\\
A Hamiltonian, $\HH$, is defined from the zeroth component of the momentum vector as
\begin{eqnarray}
\HH=\eta_0.
\end{eqnarray}
From which an integration of the characteristic $1$-form along a trajectory can be expressed as
\begin{eqnarray*}
\I(\gamma)&=&\int_\gamma\omega,\\
&=&\int_\gamma\left(\eta_i d\xi^i-\HH d\tau\right),\\
&=&\int_{\tau_0}^{\tau_1}d\tau\gamma^*
\left(\eta_i \frac{d\xi^i}{d\tau}-\HH \right).
\end{eqnarray*}
This integral is called an action, where $\gamma$ is the trajectory and $\gamma^*(\bullet)$ is a pull of $(\bullet)$ by $\gamma$. The characteristic $1$-form can be expressed as
\begin{eqnarray}
\omega&=&\left(\gamma^*\LL\right)(\tau)d\tau,\nonumber \\
\LL&=&\eta_i \frac{d\xi^i}{d\tau}-\HH. \label{lagdef}
\end{eqnarray}
Here $\LL$ is the Lagrangian. 
\end{definition}
\noindent
Since the action is independent of the parameterization of the trajectory, the pull of the trajectory $\gamma^*$ is hereafter omitted unless required for clarity. The action can be interpreted as the  ``distance'' between two points on the space-time manifold measured by the characteristic $1$-form $\omega$.
A detailed treatment of the trajectory $\gamma$ is given in Section.~\ref{ushs}.
%
%
\subsection{Examples of Dynamics}
Here we provide two concrete examples of the general dynamical space from thermodynamics and Hamiltonian mechanics. Similarity between these two examples guides us to new insights in classical dynamics.
\begin{example}{\bf (Irreversible Thermodynamics)}\footnote{These examples are given in \cite{story2002dynamics}. Other examples are provided therein.}
\vskip -1.0cm
\end{example}
\noindent
First, let us consider an isolated system of gas. The adiabatic free expansion of an isolated system enlarges the entropy $(S) ($according to the second law of thermodynamics$)$ and maintains constant temperature$ (T)$ because heat energy cannot be gained or lost. In this case, the order parameter is entropy$ (S\in\T)$, while the one-dimensional space-manifold is volume $(V\in\MM)$. The thermodynamics is described by the following characteristic function:
\begin{eqnarray}
f_{TD}&=&pV-ST\label{f5},
\end{eqnarray}
where $p$ is the pressure of gas in an insulating container and $p$ is the pressure of the system. From the thermodynamic characteristic function, $\eta_0$ can be obtained as $\partial f_{TD}/\partial S=T$. The general dynamical space then becomes
\begin{eqnarray*}
(\xi^0,\xi^1)&=&(S,V),\\
(\eta^0,\eta^1)&=&(T,p),\\
\omega_E&=&pdV-TdS,\\
\Omega_E&=&dp\wedge dV-dT\wedge dS.
\end{eqnarray*}
The characteristic function can be expressed as a scalar function on the space-time manifold as follows:
\begin{eqnarray*}
f_{TD}&=&\xi^\mu\eta_\mu,\\
&=&\xi^1\eta_1-\xi^0\eta_0,\\
&=&pV-ST.
\end{eqnarray*}
We assume that a Hamiltonian exists such that $\eta_0=T(p,V)$.
From the characteristic function and space-time manifold, we obtain the following characteristic equations:
\begin{eqnarray*}
\frac{dV}{dS}=\frac{\partial T}{\partial p},~~
\frac{dp}{dS}=-\frac{\partial T}{\partial V},~~
\frac{dT}{dS}=\frac{\partial T}{\partial S}=0.
\end{eqnarray*}
These are known as the Maxwell relations for an adiabatic transition. The characteristic $1$-form corresponds to the energy of the system ($\omega_E=-dE$)\footnote{For a statistical physics treatment, see for example \cite{landau1980statistical}}.\\
Next let us consider a system contained in an insulating container with an expandable wall and its isothermal reversible transition. The initial temperature of the system is $T_1$. The system is attached to a heat bath of temperature $T_2>T_1$ and pressure is maintained constant. Here temperature is adopted as the order parameter ($T\in\T$) because temperature increases monotonically in this case. In this example, the one-dimensional space-manifold is pressure $p\in\MM$. From the characteristic function, eq.(\ref{f5}), $\eta_0$ can be obtained as $\partial f_{TD}/\partial T=S(V,p)$, and the general dynamical space becomes
\begin{eqnarray*}
(\xi^0,\xi^1)&=&(T,p),\\
(\eta^0,\eta^1)&=&(S,V),\\
\omega_G&=&Vdp-SdT,\\
\Omega_G&=&dV\wedge dp-dS\wedge dT,
\end{eqnarray*}
and the characteristic functions are derived as
\begin{eqnarray*}
\frac{dp}{dT}=\frac{\partial S}{\partial V},~~
\frac{dV}{dT}=-\frac{\partial S}{\partial p},~~
\frac{dS}{dT}=\frac{\partial S}{\partial T}=0.
\end{eqnarray*}
These functions are alternative expressions of the above-derived Maxwell relations.
Here the characteristic $1$-form corresponds to the Gibbs free energy of the system, given as
\begin{eqnarray}
\omega_G&=&Vdp-TdS=dG.\label{HFE}
\end{eqnarray}
\begin{example}{\bf (Hamiltonian Formalism of mass points)}
\vskip -1.0cm
\end{example}
\noindent
Next let us consider the well-known Hamiltonian formalism of mass points with $n$ degrees of freedom. The order parameter and space-time manifold are defined as $t\in\T$ and $(q^1,\cdots,q^n)\in\MM$, respectively. The characteristic function is
\begin{eqnarray*}
f&=&\xi^\mu\eta_\mu,\\
&=&q^ip_i-Ht,
\end{eqnarray*}
where $\bm\xi=(t,q^1,\dots,q^n)$ and $\bm\eta=(H,p_1\,\cdots,p_n)$.
Assuming the Hamiltonian as $H=H(t,q^i,p_i)$, the characteristic forms are obtained as
\begin{eqnarray*}
\omega&=&p_idq^i-Hdt,\\
\Omega&=&dp_i\wedge dq^i-dH\wedge dt.
\end{eqnarray*}
Then, the celebrated canonical equations of motion can be expressed as the following characteristic functions: 
\begin{eqnarray*}
\frac{dq^i}{dt}=\frac{\partial H}{\partial p_i},~~
\frac{dp_i}{dt}=-\frac{\partial H}{\partial q_i},~~
\frac{dH}{dt}=\frac{\partial H}{\partial t}.
\end{eqnarray*}
\section{Underlaying Structure of Hamiltonian Systems}\label{ushs}
The similarity between thermodynamics and the Hamiltonian formalism of mass points was highlighted in the previous section. Both systems show the same symplectic structure of base manifold and evolve along an order parameter under a set of ``equations of motion.''  However, thermodynamic ideal gas systems are known to possess an underlying structure generated by the statistical mechanics of independent gas molecules. On the other hand, the Hamiltonian formalism of mass points requires no underlying structure. Here we treat the Hamiltonian formalism as a thermodynamic system and assume a virtual underlying structure for the motions of mass points. Among several candidates for a microscopic entity governing the Hamiltonian formalism, particle trajectories are adopted for the following reasons. The Hamiltonian formalism of mass points differs from ideal-gas thermodynamics primarily by importance of the trajectory. The main goal of the former is to determine the trajectory of a mass point under applied forces and initial conditions. In contrast, in the latter, the trajectory cannot be measured and has no essential meaning, similar to the quantum mechanics of mass points. By considering particle trajectories as the statistical entity, a relation between classical and quantum mechanics may be clarified.
This section considers the statistical mechanics of particle trajectories in the general dynamic space of {\bf Definition~\ref{sympmani}}.
\subsection{Geometrical Preparation}
This subsection introduces the geometrical objects used in subsequent discussions.
\begin{definition}\label{curvsp}{\bf (Curvilinear Path)}\\
A set of maps $\gamma$ such that
\begin{eqnarray*}
&\gamma& : \T \rightarrow \Gamma \subset \M : t \in [t_1,t_2] \mapsto 
{\bf \gamma}(t)=\left(\gamma^0(t)=t, \gamma^1(t), \gamma^2(t), \cdots, \gamma^n(t)\right),\label{curveg}\\
&\gamma^i&:\R \rightarrow \R:t\mapsto\gamma^i(t),~~\gamma^i\in C^\infty,
\end{eqnarray*}
is called a curvilinear path $($or simply ``path''$)$, and a set  of paths, $\Gamma$, is called a curvilinear-path space. 
\end{definition}
\noindent
In this section, we consider only those paths whose end points are fixed at $\gamma(t_1)=\bm\xi_1=(t_1,\xi_1^1,\cdots,\xi_1^n)$ and $\gamma(t_2)=\bm\xi_2=(t_2,\xi_2^1,\cdots,\xi_2^n)$.
\begin{definition}\label{velosity}{\bf (Velocity Vector)}\\
A velocity vector is a tangent vector at a point $\gamma(t)$ on the curvilinear path, defined as
\begin{eqnarray*}
\frac{d\gamma}{dt}(t)&=&\left(1,\frac{d\gamma^1}{dt}(t),\frac{d\gamma^2}{dt}(t),
\cdots,\frac{d\gamma^n}{dt}(t)\right).
\end{eqnarray*}
\end{definition}
\noindent
Hereafter, the velocity vector is written as $\dot{\bf \gamma}(t)=(1,\dot{\gamma}^1(t),\dot{\gamma}^2(t),\cdots,\dot{\gamma}^n(t))$. The velocity vector can be expressed in terms of natural bases on a tangent space $T_t\gamma$ at $\gamma(t)$ as
\begin{eqnarray*}
\dot{\bf \gamma}^i(t)=
\frac{d\gamma^i(s)}{ds}\frac{\partial}{\partial \gamma^i}\Big|_{s=t},
\end{eqnarray*}
where $i$ runs from $1$ to $n$ and the components are not summed.
A tangent vector bundle $\dot{\Gamma}=\bigcup_{\gamma(t)\subset\gamma\in\Gamma}T_t\gamma$ is called a velocity bundle.
\begin{definition}{\bf (Variational Vector)}\\
A variational vector is defined as a cotangent vector at $\gamma(t)\in\M$ on the curvilinear path $\delta\gamma(t)$.
In terms of natural bases of the cotangent space $T^*_t\gamma$, variational vectors can be expressed as
$
\delta\gamma^i(t)=\delta\gamma^i d\gamma_i
$, where $i$ runs from $1$ to $n$ and the components are not summed. The zeroth component is $\delta\gamma^0(t)=0$.
A cotangent vector bundle 
$
\delta\Gamma=\bigcup_{\gamma(t)\subset\gamma\in\Gamma}T^*_t\delta\gamma
$ is called a variational bundle. 
\end{definition}
\noindent
The bases of a velocity bundle and variational bundle are orthogonal, i.e.,
$
d\gamma^\mu\partial/\partial\gamma^\nu=\delta^\mu_\nu
$.
\begin{definition}\label{henbun}{\bf (Variational Operator)}\\
A map induced by a variational vector such that
\begin{eqnarray*}
\delta:\Gamma \rightarrow \Gamma:
\gamma \mapsto (\delta\circ\gamma)(t)=\left(\gamma+\delta\gamma \right)(t),
\end{eqnarray*}is called a variational operator. 
\end{definition}
\noindent
Here $\left(\gamma+\delta\gamma \right)(t)$ denotes the sum of two vectors $\gamma$ and $\delta\gamma$, which are defined on $\M$. The curvilinear path $\delta_\gamma(t)=(\delta\circ\gamma)(t)$ is assumed to become an element of the curvilinear-path space, i.e., $\delta_\gamma(t)\in C^\infty$. 
The distance between two paths $\gamma$ and $\delta_\gamma$ is defined as
\begin{eqnarray*}
||\delta_\gamma-\gamma||&=&||\delta\gamma||,\\
&=&\int_{t_1}^{t_2}dt|\delta\gamma(t)|,\\
&=&\int_{t_1}^{t_2}dt\left\{\sum_{k=1}^n\left(\delta\gamma^k(t)\right)^2\right\}^{1/2}.
\end{eqnarray*}
Suppose that $\f$ is a functional defined on $\Gamma$ such that $\f:\Gamma \rightarrow \R^m:\gamma \mapsto \f(\gamma),~m\in\N$. 
The variational operator maps a functional $\f$ to another functional as
$
\delta:\f(\gamma) \mapsto \delta\f(\gamma)=\f(\delta_\gamma).
$
Since $\gamma$ is a map defined in {\bf Definition \ref{curvsp}}, a functional $\f$ can be pulled back to a function defined on $\R$ using a pull of $\gamma$, which denoted as $\gamma^*$:
\begin{eqnarray*}
\gamma^*:\f(\gamma) \mapsto \gamma^*\left(\f(\gamma)\right)(\tau)\in\R.
\end{eqnarray*}
For simplicity, we use a shorthand, 
$\gamma^*(\f(\gamma))(\tau)=\f(\gamma)(\tau)$.
A variational operator can be pulled as
$
\gamma^*(\delta_\gamma\f(\gamma))(\tau)=\f(\gamma+\delta\gamma)(\tau).
$
Thus variation of a functional $\f$ can be defined as
\begin{eqnarray*}
\frac{\delta\parallel\f(\gamma)\parallel}{\delta\gamma}&=&\displaystyle \lim_{\parallel\delta\gamma\parallel \to 0}
\frac{\parallel\f(\gamma+\delta\gamma)\parallel-\parallel\f(\gamma)\parallel}
{\parallel\delta\gamma\parallel}.
\end{eqnarray*}
When a variation is zero with some $\gamma_c$, it is said that the functional $\f$ has a extremal at $\gamma_c$.
\subsection{Dynamics of Paths}
We now define the general dynamic space occupied by a mass point and impose a probability space on it. For simplicity, we treat a single mass point.
\begin{definition}\label{Laglangespace}{\bf (General Dynamic Path Space)}\\
The general dynamical space occupied by a point particle is described as follows:
\end{definition}
\begin{itemize}
\item {\bf Space Manifold}: The space manifold $\MM$ is a three-dimensional Euclidean space $\R^3$. 
\item {\bf Time Manifold}: The time manifold $\T$ is Newtonian absolute time, which is commonly used in inertial system analysis. A space-time manifold $\T\otimes\MM=\R\otimes\R^3$ is called a Galilean Manifold.
\item {\bf Characteristic Function}: The characteristic function (functional) is defined as 
\begin{eqnarray*}
f&=&\pi_\mu\gamma^\mu,\\
&=&\pi_i\gamma^i-\HH({\bm \pi},{\bm \gamma})t, 
\end{eqnarray*}
where ${\bm \gamma}$ is a path defined on $\MM$, and ${\bm \pi}$ is a vector defined on the velocity bundle. The mass is defined as $m=||{\bm \pi}||/||{\bm {\dot \gamma}}||$. The characteristic function $f$ can be considered as a functional of the path defined on $\Gamma\otimes{\dot \Gamma}$.
\item {\bf Characteristic Forms}: The characteristic 1- and 2-forms, respectively, are defined as
\begin{eqnarray*}
\omega&=&\pi_\mu d\gamma^\mu,\\
&=&\pi_i d\gamma^i-\HH dt,\\ 
\end{eqnarray*}
and
\begin{eqnarray*}
\Omega&=&d\pi_\mu\wedge d\gamma^\mu,\\
&=&d\pi_i\wedge d\gamma^i-d\HH\wedge dt^\mu.
\end{eqnarray*}
\end{itemize}
Here the characteristic functional is assumed to be invariant under affine transformation on the space manifold\footnote{If the characteristic functional is assumed to be invariant under affine transformation on the entire space-time manifold, then the manifold is called a {\it Minkowski manifold} and the system is called {\it relativistic}.}, $\MM$. In this notation, the action and Lagrangian, respectively, are expressed as 
\begin{eqnarray*}
\I(\gamma)&=&\int\LL dt,
\end{eqnarray*}
and
\begin{eqnarray*}
\LL dt&=&(\pi_i d\dot\gamma^i-\HH)dt,\\
&=&\pi_\mu d\gamma^\mu.
\end{eqnarray*}
This transformation from the Hamiltonian to Lagrangian is known as a Legendre transformation and the independent variables of $\LL$ are now $(\gamma,\dot\gamma)$. Here the flat Riemannian metric $g_{\mu\mu}=(-,+,+,+)$ correctly induces a Legendre transformation from $\HH$ to $\LL$.\\
From the geometrical framework introduced in the previous subsection, we now construct a dynamical path system of a mass point evolving under the Hamiltonian formulation.
\begin{theorem}\label{Hamilton}{\bf (Hamilton)}\\
In a general dynamical space, the following two trajectories are equivalent:
\begin{enumerate}
\item Trajectory that gives the extremal of variation of the action: $\delta\I(\gamma)=0$.\\
\item Trajectory that satisfies the characteristic equations:
\begin{eqnarray*}
\frac{d\gamma^i}{dt}=\frac{\partial\HH}{\partial \pi_i},~~
\frac{d\pi_i}{dt}=-\frac{\partial\HH}{\partial \gamma_i},~~
\frac{d\HH}{dt}=\frac{\partial\HH}{\partial t}.
\end{eqnarray*}
\end{enumerate}
\end{theorem}
\begin{proof}
$1 \Rightarrow 2$: Applying a variational operator to the action, we obtain
\begin{eqnarray*}
\delta\I(\gamma)&=&\delta\int_\gamma\omega,\\
&=&\int_{\delta\gamma}\omega=0,
\end{eqnarray*}
implying that the integration of the characteristic 1-form is independent of $\delta\gamma$ and satisfies Principle \ref{calpri}. Then, the Hamiltonian must satisfy the characteristic equations. \\
$2 \Rightarrow 1$: Applying the variational operator to the Lagrangian, we obtain
\begin{eqnarray*}
\delta\LL&=&\delta(\pi_i\dot{\gamma}^i)-\delta_\HH,\\
&=&\delta\pi_i\dot{\gamma}^i+\pi_i\delta\dot{\gamma}^i
-\left(\frac{\partial \HH}{\partial \gamma^i}\delta\gamma^i+\frac{\partial \HH}{\partial \pi_i}\delta\pi_i\right),\\
&=&\delta\pi_i\dot{\gamma}^i+\pi_i\delta\dot{\gamma}^i
-\left(-\dot{\pi}_i\delta\gamma^i+\dot{\gamma}^i\delta\pi_i\right),\\
&=&\delta\pi_i\dot{\gamma}^i+\pi_i\delta\dot{\gamma}^i
-\left(\pi_i\delta\dot{\gamma}^i+\dot{\gamma}^i\delta\pi_i\right)=0,
\end{eqnarray*}
where $i=1,2,3$. Here we use integration by parts and assume zero variations at both ends of the path.
Steps 2 and 3 in the derivation are obtained by substituting the characteristic equations.
\end{proof}
\noindent
The variation of the action, on the other hand, can be written as $\delta\I(\gamma)=\int dt\delta\LL$, which yields $\delta\I(\gamma)=0~\Rightarrow~\delta\LL=0$. This is analogous to the extremal of the Gibbs free energy in an isothermal reversible system at equilibrium. Let us consider this analogy in more detail. As pointed out, when introducing eq.(\ref{HFE}), the characteristic 1-form of the isothermal reversible transition is equivalent to the Gibbs free energy. In this analogy, the Hamiltonian corresponds to the entropy of the system and the macroscopic system configuration may be determined by the extremal point of the characteristic 1-form, which corresponds to the Lagrangian. The analogies between the Hamiltonian formalism of mass points and thermodynamics of ideal gases is summarized in Table~\ref{tbl1}.
\begin{table}[t]
\begin{center}
 \caption{\label{tbl1}
Analogies between isothermal reversible thermodynamics and statistical dynamics of particle paths.
}
   \begin{tabular}{lccc}
~& adiabatic  &
isothermal &statistical dynamics \\ 
~&free expansion&reversible transition&of paths \\ \hline
{\bf Order parameter} & entropy:$S$&temperature:$T$&time:$t$ \\
{\bf Momentum vector}&pressure:$p$& volume:$V$ & momentum:$\pi$ \\
{\bf Space coordinate}&volume:$V$ & pressure:$p$& position:$\gamma(t)$ \\
{\bf Hamiltonian}&temperature:$T$&entropy:$S$ & Hamiltonian:$\HH$\\
{\bf Lagrangian}&internal energy:&Gibbs free energy: & Lagrangian: \\ 
~&$-dU=p~dV-TdS$&$dG=Vdp-SdT$ & $\LL dt=\pi d\gamma-\HH dt$ \\
   \end{tabular}
\end{center}
\end{table}
This analogy will be pursued further in the next subsection.
\subsection{Statistical Mechanics of Trajectories}
The curvilinear path defined in the previous subsection is the trajectory of the particle. This trajectory is considered as the microscopic basis for constructing a statistical mechanical analog of the thermodynamic system.
\begin{definition}\label{lag1}{\bf (Lagrangian Probabilistic Space)}\label{psp}\\
The Probabilistic space $\{\Gamma,{\mathfrak P}(\Gamma),p(\gamma) \}$ is imposed on the general dynamic path space $($see {\bf definition \ref{Laglangespace}}$)$ as follows:
\end{definition}
\begin{enumerate}
\item {\bf Whole event} The set of whole events is the curvilinear-path space.
\item {\bf $\sigma$-algebra}: The $\sigma$-algebra  is a power set of $\Gamma$, denoted ${\mathfrak P}(\Gamma)$.
\item {\bf Probability Measure}: Consider a map such that
\begin{eqnarray*}
p:\Gamma\rightarrow[0,1]\in\R:\gamma\mapsto p(\gamma),
~~\sum_{\gamma\in\Gamma}p(\gamma)=1.
\end{eqnarray*}
\end{enumerate}
This functional describes the probability density to realize the path $\gamma$. Please note that we consider only those paths whose start and end points are fixed at $\bm\xi_1$ and $\bm\xi_2$, respectively. The initial momentum vector must be chosen to realize a classical path.
The probabilistic space $\{\Gamma,{\mathfrak P}(\Gamma),p(\gamma) \}$ is assumed to satisfy the probability axioms proposed by Kolmogorov\cite{Kolmogorov}. If a set $\Gamma$ with infinite degrees of freedom belongs to $\aleph_1$, then a set ${\mathfrak P}(\Gamma)$ must belong to $\aleph_2$. Since we cannot mathematically justify the probability measure defined on such an infinite set, we present a formal treatment only. The existence of the measure can be verified in limited cases, as will be discussed later.
\begin{definition}\label{EntpyPath}{\bf (Entropy of Paths)}\\
The entropy of the paths is defined as
\begin{eqnarray}
{\cal S}&=&-\sum_{\gamma\in\Gamma}p(\gamma)\log{p(\gamma)}\label{EntpyEq},
\end{eqnarray}
according to Shannon{\rm \cite{Shannon1948}}. Here $\Gamma$, $\gamma\in\Gamma$ and $p(\gamma)$ are defined in the {\bf Definition~\ref{psp}}
\end{definition}
\noindent
Consistency between the above definition of entropy and Hamiltonian formalism will be discussed later. At this stage, we lack detailed knowledge of the dynamics that govern path behavior. However, it seems natural to configure paths by the following principle.
\begin{principle}\label{MEP}{\bf (Maximum Entropy Principle)}\\
Path configuration is determined to maximize the entropy of the paths.
\end{principle}
\noindent
According to above principle, a probability $p(\gamma)$ to observe a path $\gamma$ can be given by following theorem.
\begin{theorem}\label{CEB}{\bf (Canonical Ensemble)}\\
The probability $p(\gamma)$ defined as following two requirement are equivalent:
\begin{enumerate}
\item The path whose entropy is maximized under the following constraints
\begin{eqnarray}
\sum_{\gamma\in\Gamma}p(\gamma)=1\label{CE1},\\
\sum_{\gamma\in\Gamma}p(\gamma)\I(\gamma)=I\label{CE2}.
\end{eqnarray}
The first constraint is conservation of probability. The second stipulates that an action averaged over all possible paths $I$, called a {\it classical action}, must exist.
\item The path whose probability $p(\gamma)$ is given as
\begin{eqnarray}
p(\gamma)&=&
\frac{\exp{\left(-\beta\I(\gamma)\right)}}{Z(\beta)}\label{CE3},\\
Z(\beta)&=&
\sum_{\gamma\in\Gamma}\exp{\left(-\beta\I(\gamma)\right)},\label{CE4}
\end{eqnarray} 
where $\beta$ is a constant to eliminate a dimension in the argument of exponential function.
\end{enumerate}
\end{theorem}
\noindent
Here, the functional integration $Z(\beta)$ is regarded as a partition functional by analogy with equilibrium statistical mechanics. Then, particle trajectories (paths) form a canonical ensemble.
\begin{proof}
To maximize the entropy of the constrained paths, we introduce Lagrange multipliers $\alpha$ and $\beta$  such that
\begin{eqnarray*}
\phi(p,\alpha,\beta)&=&
-S
+\alpha\left(\sum_{\gamma\in\Gamma}p(\gamma)-1\right)
+\beta\left(\sum_{\gamma\in\Gamma}p(\gamma)\I(\gamma)-I\right).
\end{eqnarray*}
The following conditions must be satisfied:
\begin{eqnarray*}
\frac{\partial\phi(p,\alpha,\beta)}{\partial\alpha}&=&
\sum_{\gamma\in\Gamma}p(\gamma)-1=0,\\
\frac{\partial\phi(p,\alpha,\beta)}{\partial\beta}&=&
\sum_{\gamma\in\Gamma}p(\gamma)\I(\gamma)-I=0,\\
\frac{\partial\phi(p,\alpha,\beta)}{\partial p}&=&
\log{p(\gamma)}+1+\alpha+\beta\I(\gamma)=0.
\end{eqnarray*}
Here we used the functional derivative rules presented in the Appendix.
Solving the above equations gives
\begin{eqnarray*}
p(\gamma)&=&\frac{\exp{\left(-\beta\I(\gamma)\right)}}
{\sum_{\gamma\in\Gamma}\exp{\left(-\beta\I(\gamma)\right)}}.
\end{eqnarray*}
To ensure that $\exp{\left(-\beta\I(\gamma)\right)}$ converges when $\I(\gamma)\rightarrow\infty$, $\beta$ must be positive.
\end{proof}
\noindent
The above theorem implies that introducing the entropy of paths described by eq. (12) yields the canonical ensemble of equilibrium statistical mechanics. Therefore, it appears that the classical trajectory of a mass point can be interpreted as the equilibrium state among all possible paths. The following theorem should then naturally hold:
\begin{theorem}\label{EoM}{\bf (the most probable path)}\\
The following two trajectories are equivalent:
\begin{enumerate}
\item Trajectory that gives the extremal of variation of the action: $\delta\I(\gamma)=0$.
\item Trajectory that gives the maximum probability in eq.$($\ref{CE3}$)$.
\end{enumerate}
\end{theorem}
\begin{proof}
Applying the variational operator $\delta$ to both sides of eq.(\ref{CE3}), we obtain
\begin{eqnarray*}
\delta p(\gamma)&=&-\beta\frac{\exp{\left(-\beta\I(\gamma)\right)}}{Z}\delta\I(\gamma).
\end{eqnarray*}
Thus, $\delta p(\gamma)=0$ and $\delta\I(\gamma)=0$ are equivalent. Since $\beta$, $Z$, and $\exp{(-\beta\I(\gamma))}$ are positive, the path $\gamma$ that minimizes $\I(\gamma)$ gives the maximum $p(\gamma)$.
\end{proof}
\noindent
Above theorem posits that the classical trajectory described in{\bf Theorem \ref{Hamilton}} is the most probable path of a mass point under {\bf Principle \ref{MEP}}. Thus, we have proved that the maximum entropy principle ({\bf Principle \ref{MEP}}) and Cartan principle ({\bf Principle \ref{calpri}}) are mathematically equivalent in a general dynamic space.
\section{Classical Mechanics to Quantum Mechanics}
This section relates the formulations of classical and quantum mechanics. Quantum and classical mechanics  differ most distinctly by the probability amplitude, whose square gives probability density. Quantum mechanical  motions, embodied in the Heisenberg/Schr\"{o}dinger equation, are governed by the time evolution of probability amplitudes rather than probability densities. Importance of the quantum probability amplitude is discussed elsewhere\cite{2013arXiv1304.5824K}. The following subsection introduces a general framework of quantum mechanics, within which we relate our formulation to path integrals and stochastic quantization.
\subsection{General Quantum System}
Here we develop a general framework for defining quantum mechanical probability amplitudes. 
\begin{definition}\label{gcy}{\bf (Quantum Amplitude\footnote{In a narrow sense,``quantum amplitude'' is a complex number whose square of the absolute value is a probability (density). In this report, we use a word ``quantum amplitude'' not only for complex numbers, but also for vectors whose square of the absolute value is a probability (density).} and Probability)}\\
Let $\K$ be any field, not necessarily commutative, and let $V$ be a linear (vector) space on it. The base field $\K$ is associated with each element of $\Gamma$. The probability amplitude and probability measure are introduced on these spaces as follows:
\begin{enumerate}
\item  We introduce the following map from the path to an n-dimensional vector space:
\begin{eqnarray*}
\psi^n_\K:\Gamma\rightarrow V^n=\underbrace{V\otimes\cdots\otimes V}_{n}:\gamma\mapsto\psi^n_\K(\gamma),
\end{eqnarray*}
where the paths have fixed end points at $\bm\xi_1$ and $\bm\xi_2$ {\rm (}see {\bf Definition \ref{curvsp}}{\rm )}. The amplitude is then defined as
\begin{eqnarray*}
\psi^n_\K(\bm\xi_2;\bm\xi_1)&=&\sum_{\gamma\in\Gamma}\psi^n_\K(\gamma).
\end{eqnarray*}
In defining a measure on the vector space, a $\sigma$-algebra of the probability amplitudes in the base field $\K$ is assumed. Hereafter, the simplified notation 
$
\psi^n_\K
$
denotes that the start and end points are fixed at $\bm\xi_1$ and $\bm\xi_2$, respectively, unless otherwise stated. 
\item The following map from an amplitude to a real number 
\begin{eqnarray*}
\mu:V^n\rightarrow\R:\psi^n_\K
\rightarrow \mu\left(\psi^n_\K\right)\in[0,1],
\end{eqnarray*}
is called a quantum probability (measure). The sequential map 
\begin{eqnarray*}
\mu\circ\psi^n_\K:\Gamma\rightarrow(V^n\rightarrow)\R:\Gamma\mapsto \mu
=\mu\left(\Gamma\right)=\mu(\psi^n_\K),
\end{eqnarray*}
is also called a quantum probability and is represented by the same symbol $\mu(\gamma)$, where $\K$ and $V$ are as previously defined. The quantum probability measure must satisfy
$
\mu(\psi^n_\K)\leq 1.
$
This measure is not normalized to unity because the curvilinear-path space $\Gamma$ includes only paths with fixed end points; however, the quantum mechanical uncertainty relation precludes precise determination of an end point.
\end{enumerate}
\end{definition}
\noindent
Two essential differences exist between the above-described general quantum system and canonical ensemble introduced earlier:
\begin{enumerate}
\item Relaxation of the first constraint in {\bf Theorem \ref{CEB}.1}. 
\item The probability amplitude is not necessarily a real number.
\end{enumerate}
These differences may lead the dynamical system to adopt quantum mechanical instead of classical behavior.
\subsection{Path Integral Quantization}
Here we derive the probability measure and amplitude from the maximum entropy principle ({\bf Principle \ref{MEP}}). The probability that a mass point observed at ${\bm \xi_1}$ is later observed at ${\bm \xi_2}$ is  
\begin{eqnarray}
p(\bm \xi_2,\bm \xi_1)=\mu(\Gamma)
=\Bigl|\sum_{\gamma\in\Gamma}\psi^1_\C(\gamma)\Bigr|^2\label{trans},
\end{eqnarray}
where the quantum probability amplitude $\psi^1_\C(\gamma)$ resides in a one-dimensional vector space on a complex number field $\C$. Hereafter, $\psi^1_\C(\gamma)$ is written as $\psi(\gamma)$ for simplicity. 
\begin{theorem}\label{pathint}{\bf (Path Integral Quantization)}\\
The quantum probability amplitude and probability measure that minimize the entropy 
\begin{eqnarray}
S&=&-\sum_{\gamma\in\Gamma}\Bigl|\psi(\gamma)\Bigr|^2
\log{\Bigl|\psi(\gamma)\Bigr|^2}\label{EntpyEq2},
\end{eqnarray}
under constraints 
\begin{eqnarray}
\sum_{\gamma\in\Gamma}\Bigl|\psi(\gamma)\Bigr|^2\I(\gamma)&=&I\label{CE5},\\
\sum_{\gamma\in\Gamma}\Bigl|\psi(\gamma)\Bigr|^2&=&1,
\label{PathI1}
\end{eqnarray}
is given by
\begin{eqnarray}
\psi(\gamma)&\simeq&C~e^{\frac{i}{\hbar}\I(\gamma)}\label{PathI2},
\end{eqnarray}
where $C\in\R$ is an appropriate normalization constant and $\I$ is rendered dimensionless by dividing by the constant $\hbar$.
\end{theorem}
\begin{proof}
The probability amplitude that maximizes the entropy is again obtained by the Lagrange multiplier method:
\begin{eqnarray*}
\phi(\psi,\alpha,\beta)&=&-S+\alpha\left(
\sum_{\gamma\in\Gamma}\Bigl|\psi(\gamma)\Bigr|^2-1
\right)
+\beta
\left(
\sum_{\gamma\in\Gamma}\Bigl|\psi(\gamma)\Bigr|^2\I(\gamma)-I
\right),\\
\frac{\partial\phi(\psi,\alpha,\beta)}{\partial\alpha}
&=&\sum_{\gamma\in\Gamma}\Bigl|\psi(\gamma)\Bigr|^2-1=0,\\
\frac{\partial\phi(\psi,\alpha,\beta)}{\partial\beta}
&=&\sum_{\gamma\in\Gamma}\Bigl|\psi(\gamma)\Bigr|^2\I(\gamma)-I=0,\\
\frac{\partial\phi(\psi,\alpha,\beta)}{\partial\psi^*(\gamma)}&=&
\psi(\gamma)\left(\log{\Bigl|\psi(\gamma)\Bigr|^2}+1
+\alpha+\beta\I(\gamma)\right)=0.
\end{eqnarray*}
Here we again used rules of a functional derivative described in Appendix.
Solving the above equations, the probability amplitude is obtained as
\begin{eqnarray*}
\psi(\gamma)&=&e^{-\frac{1}{2}\beta\I(\gamma)}/Z^{1/2},\\
Z&=&\sum_{\gamma\in\Gamma}e^{-\Re[\beta]\I(\gamma)}.
\end{eqnarray*}
Here $\beta\in\C$ and $\Re[\beta]$ ($\Im[\beta]$) are real-part (imaginary) of $\beta$, respectively. An imaginary part of $\beta$ is arbitrary due to an $U(1)$ symmetry of $\psi(\gamma)$.
As shown in the {\bf Theorem \ref{EoM}}, a main contribution for the probability comes from the trajectory of the classical mass point. Moreover it is expected that a quantum mechanical path is small fluctuated around this classical trajectory.
Thus here we assume that a real-part of $\beta$ contributes to $\psi(\gamma)$ as slow moving function and replace it by a mean value. Then the probability amplitude can be express as
\begin{eqnarray*}
\psi(\gamma)&=&\frac{1}{Z^{1/2}}e^{-\frac{1}{2}\beta\I(\gamma)},\\
&=&\frac{1}{Z^{1/2}}
e^{-\frac{1}{2}\left(\Re[\beta]+i\Im[\beta]\right)\I(\gamma)},\\
&\simeq&C e^{-\frac{i}{\hbar}\I(\gamma)}.
\end{eqnarray*}
Where $\hbar=2\Im[\beta]^{-1}$ and
\begin{eqnarray*}
C&=&\frac{1}{Z^{1/2}}
\sum_{\gamma\in\Gamma}e^{-\frac{1}{2}\Re[\beta]\I(\gamma)}.
\end{eqnarray*}
\end{proof}
\noindent
Eqs.(\ref{PathI2}) and  (\ref{trans}) are nothing but the path-integral representation of transition probability introduced by Feynman\cite{feynman1942feynman,feynman1965quantum}\footnote{See also a section 1.3 of \cite{kaku1999introduction}}. The constant $\hbar$ is not necessarily the Plank constant and cannot be determined within this formulation. Instead, a transition from classical to quantum mechanics arises through the probability amplitude, which can be a complex valued functional in our interpretation, in contrast to the real probability density of classical mechanics. This transition becomes evident if eqs.$(\ref{EntpyEq})$ and $(\ref{CE2})$ are compared with eqs.$(\ref{EntpyEq2})$ and $(\ref{CE5})$.
\section{Conclusions and Discussions}
In this report, we introduced a general dynamic space that allows a unified geometric viewpoint of various dynamic systems. System dynamics were geometrically introduced though Cartan's principle. The equations of motion derived from Cartan's principle were found to be mathematically equivalent to Hamiltonian dynamics. Under the proposed generalized framework, the dynamics of a mass point were equivalent to those of equilibrium thermodynamics, enabling the derivation of a thermodynamic analogue of mass point dynamics. In fact, the maximum entropy principle defined in trajectory space generated precisely the Hamiltonian equation of motion. The classical trajectory of a mass point can be interpreted as the most probable path of the point. By extending the maximum entropy principle to probability amplitude rather than probability density, we retrieved the equations of path-integrated quantum mechanics. The probability amplitude was essential for transferring the system from a classical to quantum state. In summary, we incorporated various dynamical systems such as classical mechanics of a mass point and equilibrium thermodynamics and quantum mechanics of a point particle into a general mathematical framework. \\
While this framework provides a unique vantage point for both classical and quantum mechanics, it is not yet suitable for quantum mechanical analysis. We defined quantum amplitudes on a curvilinear space of precisely fixed end points. A basic quantum mechanical element is not naturally located in such a space-time manifold because of violation of the uncertainty relation. To satisfy the uncertainty relation, the  essential element of quantum mechanics must be defined on a measurable space. A suitable candidate manifold is the Cartan tube introduced in Section\ref{DYonSyMan}, which is defined in a measurable extended phase space. The elements of this space, space-time and momentum manifolds comprise a Fourier-dual pair\cite{kurihara}. Thus, we can expect to construct quantum mechanics that satisfy the uncertainty condition on the Cartan tube. Moreover this formalism is suitable to treat quantum field theory, which is considered as a more fundamental theory than quantum mechanics of a mass point. A detailed analysis of this subject is beyond the scope of this report and will be reported elsewhere.
\section*{Acknowledgments}
We wish to thank to Dr. Y.~Sugiyama, Profs. T.~Kon, and J.~Fujimoto for their continuous encouragement and fruitful discussions. The authors would like to thank Enago (www.enago.jp) for the English language review.

%
%
\newpage
\appendix
\section{Algebraic functional calculus}
The functional differential and integral calculus used in this report is not directly extendable to  an infinite dimensional space. However, an algebraic treatment of the functional derivatives required in this report is sufficient.
\subsection{Algebraic differentiation}\label{algdif}
\begin{definition}{\bf (Algebraic differentiation($1$-parameter)}\\
Algebraic differentiation is a map $Z$ from the vector, defined on an open neighborhood $U$ about the point $p$ on manifold $M$, to a complex number. The map $Z$ must satisfy three algebraic conditions:
\begin{itemize}
\item Unity:~$Z(x)=1$,
\item Linearity:~$Z(a f+b g)=a Z(f)+b Z(g)$,
\item Leibniz's rule:~$Z(g f)=f~Z(g)+Z(f)~g$,
\end{itemize}
where $x$ is a local coordinate on $U$, $f,g\in V$, and $a,b \in \R$. 
\end{definition}
A differential operator acting on local coordinates $x$ on $U$ is written as $Z=\frac{d}{dx}$ 
\begin{example}{\bf (Constant functions)}\\
Given that
\begin{eqnarray*}
&~&\frac{d}{dx}(1)=\frac{d}{dx}(1 \times 1)
=1 \times \frac{d}{dx}(1)+\frac{d}{dx}(1) \times 1=2\frac{d}{dx}(1), \\
&\Rightarrow& \frac{d}{dx}(1)=0,
\end{eqnarray*}
and that the derivative of the constant function $f=1$ is zero, it immediately follows that the derivative of any constant function is zero.
\end{example}
\begin{example}{\bf (Polynomials)}\\
The derivative of $f(x)=x^2$ is
\begin{eqnarray*}
&~&\frac{d}{dx}(x^2)=\frac{d}{dx}(x \times x)=x\frac{d}{dx}(x)+\frac{d}{dx}(x)x
=2x.
\end{eqnarray*}
Then, by mathematical induction on $n$, $\frac{d}{dx}(x^n)=n x^{n-1}$ for any $\Z\ni n\neq 0$ .
\end{example}
\begin{example}{\bf(Power function)}\\
The derivative of $f^2(x)$ is
\begin{eqnarray*}
\frac{d}{dx}(f^2)&=&\frac{d}{dx}(f \times f)
=\frac{d}{dx}(f) \times f+f \times \frac{d}{dx}(f)=2f \frac{d}{dx}(f).
\end{eqnarray*}
Then, by mathematical induction on $n$, $\frac{d}{dx}(f^n)=n f^{n-1}\frac{d}{dx}(f)$ for any $n\neq 0$
\end{example}
\begin{example}{\bf (Exponential function)}\\
An exponential function $\exp(x)$ is defined as an identity function of the differentiation operator. It is lower-bounded by $\exp(0)=1$ such that
\begin{eqnarray*}
\frac{d}{dx}(\exp(x))=\exp(x),~\exp(0)=1.
\end{eqnarray*}
On the other hand, the infinite series
\begin{eqnarray*}
{\rm Exp}(x)=\sum_{k=0}^\infty \frac{1}{k!} x^k,
\end{eqnarray*}
satisfies the same differential equation and boundary condition. Thus, $\exp(x)$ is equivalent to ${\rm Exp}(x)$ and both are hereafter expressed as ${\rm exp}(x)$.
\end{example}
\begin{example}\label{Dlog}{\bf (Logarithmic function)}\\
The logarithmic function is the inverse of the exponential function, such that
$
\exp(\log(x))=\log(\exp(x))=x.
$
Differentiating the right-hand side of the formula, we get
\begin{eqnarray*}
\frac{d}{dx}(\exp(\log(x)))=\frac{d}{dx}(\log(x)) \exp(\log(x))
=\frac{d}{dx}(\log(x))x.
\end{eqnarray*}
On the other hand, differentiating the left-hand side yields
$
\frac{d}{dx}(x)=1.
$
Equating these, we obtain
$
\frac{d}{dx}(\log(x))=\frac{1}{x}~(x\neq 0).
$
.
\end{example}
\subsection{Integration: inverse operation of differentiate}
\begin{definition}{\bf (Primitive function)}\\
A function $F(x)$ satisfying
\begin{eqnarray*}
\frac{d}{dx}(F(x))=f(x),
\end{eqnarray*}
is called a primitive function.
A map homologizing a function $f(x)$ to a primitive function $F(x)$, i.e.,
\begin{eqnarray*}
\int dx:V \rightarrow V:f(x) \mapsto \left(\int dx \right)f(x)=F(x)
\end{eqnarray*}
is called an integration.
\end{definition}
An operator $\int dx$ maps a function to its primitive function.
\begin{definition}{\bf (Definite integral)}\\
The map 
\begin{eqnarray*}
\int_\bullet^\bullet dx:V \otimes \R \otimes \R
\rightarrow \R:(f(x),\{a,b\}) \mapsto \left(\int_a^b dx \right)f(x)=
F(b)-F(a)
\end{eqnarray*}
is called a definite integral.
\end{definition}
Leibniz's rule gives rise to the following theorem:
\begin{theorem}{\bf (Integration by parts)}
\begin{eqnarray*}
\int_a^b dx \left(f(x) g(x)\right)=
\left[F(x) g(x) \right]_a^b -  \int_a^b dx \left( F(x) \frac{d}{dx}g(x)\right),
\end{eqnarray*}
where $F(x)$ is a primitive function of $f(x)$.
\end{theorem}
\begin{example}~\\
As an example, we integrate the function $f(x)=x \exp(-x^2)$ by parts. For $^\forall c \in \R$, we have 
\begin{eqnarray*}
&~&\frac{d}{dx}\left(-\frac{1}{2}\exp(-x^2)+c\right)=x \exp(-x^2),\\
&~&\Rightarrow F(x)=\int dx~f(x)=-\frac{1}{2}\exp(-x^2)+c
\end{eqnarray*}
Performing the definite integration, we obtain
\begin{eqnarray*}
\int_{-\infty}^\infty dx~f(x)=-\frac{1}{2}\exp(-x^2) \Bigl|_{x=+\infty}-\left(-\frac{1}{2}\exp(-x^2) \Bigl|_{x=-\infty} \right)=0.
\end{eqnarray*}
Here convergence of the limit can be confirmed by expressing $exp(x)$ as an infinite series.
\end{example}
\begin{example}\label{DiracDelta}{\bf (Dirac $\delta$ function)}\\
The Heaviside unit function is defined as
\begin{eqnarray*}
\Theta(x)=
\begin{cases}
1&x\geq0\\
0&x<0.
\end{cases}
\end{eqnarray*}
The Dirac delta function is formally defined as
$
\delta(x)=\frac{d}{dx}\Theta(x).
$
For any function $f(x)$, we have
\begin{eqnarray*}
\left(\int_{-\infty}^\infty dx \right)\left(\delta(x) f(x)\right)&=&
\left[ \Theta(x)f(x) \right]_{-\infty}^{\infty}-\int_{-\infty}^\infty dx 
\left( \Theta(x) \frac{d}{dx}f(x) \right)\\
&=&-f(x) \bigl|_{x=\infty}-\int_{0}^\infty dx \frac{d}{dx}f(x)\\
&=&-f(x) \bigl|_{x=\infty}-\left[ f(x) \right]_0^\infty \\
&=&f(0),
\end{eqnarray*}
where we have used integration by parts.
\end{example}
\subsection{Functional calculus}~\\
In this report, differential calculus is applied on functionals. The calculus is treated algebraically and the operations are not checked for convergence.
\begin{definition}\label{funcdevdef}{\bf (Algebraic functional calculus)}\\
A functional differential $\frac{\delta}{\delta \phi(y)}$ is a linear operation that satisfies Leibniz's rule and
\begin{eqnarray*}
\frac{\delta \phi(x)}{\delta \phi(y)}=\delta(x-y).
\end{eqnarray*}
%
\end{definition}
Elementary functionals and their functional derivatives are defined and calculated following the methods of {\bf Appendix {\it A.1}}.
\begin{example}{\bf (Exponential functional)}\\
An exponential functional is an identity function of the differential operator. It is bounded by $\exp(\phi_0)=1$, where $\phi_0=0$.The exponential functional is therefore equivalent to
\begin{eqnarray*}
{\rm Exp}[\phi(x)]=\sum_{k=0}^\infty \frac{1}{k!} \phi^k(x).
\end{eqnarray*}
From the definitions of functional calculus, we obtain
\begin{eqnarray*}
\frac{\delta}{\delta \phi(y)}\exp{\left(\int dx~\phi(x) \varphi(x)\right)}&=&
\frac{\delta}{\delta \phi(y)}
\sum_{k=0}^\infty \frac{1}{k!} \left(\int dx~\phi(x) \varphi(x)\right)^k\\
&=&\sum_{k=0}^\infty \frac{1}{k!}i
\left(\int dx~\phi(x) \varphi(x)\right)^{k-1}
\varphi(y)\\
&=&\varphi(y)\exp{\left(\int dx~\phi(x) \varphi(x)\right)},
\end{eqnarray*}.
\end{example}

\begin{thebibliography}{10}

\bibitem{arnol1989mathematical}
V.I. Arnol'd, K.~Vogtmann, and A.~Weinstein.
\newblock {\em Mathematical Methods of Classical Mechanics}.
\newblock Graduate Texts in Mathematics. Springer, 1989.

\bibitem{caltan}
E.~Cartan.
\newblock Le{\c c}ons sur les invariants int{\' e}graux.
\newblock 1922.

\bibitem{feynman1942feynman}
R.P. Feynman and L.M. Brown.
\newblock {\em Feynman's thesis: a new approach to quantum theory}.
\newblock World Scientific Publishing Company, Incorporated, 1942.

\bibitem{feynman1965quantum}
R.P. Feynman and A.R. Hibbs.
\newblock {\em Quantum Mechanics and Path Integrals}.
\newblock International series in pure and applied physics. McGraw-Hill, 1965.

\bibitem{flanders1963differential}
H.~Flanders.
\newblock {\em Differential Forms With Applications to the Physical Sciences}.
\newblock Dover Books on Mathematics. Dover Publications, 1963.

\bibitem{QTD}
M.~Gemmer, J.~Michel and G.~Mahler.
\newblock {\em Quantum thermodynamics : emergence of thermodynamic behavior
  within composite quantum systems}.
\newblock Springer, 2nd ed edition, 2009.

\bibitem{kaku1999introduction}
M.~Kaku.
\newblock {\em Introduction to Superstrings and M-Theory}.
\newblock Graduate Texts in Contemporary Physics. Springer New York, 1999.

\bibitem{Kolmogorov}
A.~N. Kolmogorov.
\newblock {\em Fundation of probability theory}.
\newblock Moscow -Leningrad, January 1936.
\newblock (in Russian).

\bibitem{kurihara}
Y.~Kurihara.
\newblock Classical information theoretic view of physical measurements and
  generalized uncertainty relations.
\newblock {\em Journal of Theoretical and Applied Physics}, 7(1):28, 2013.

\bibitem{2013arXiv1304.5824K}
Y.~{Kurihara} and N.~M.~U. {Quach}.
\newblock {Advantages of the probability amplitude over the probability density
  in quantum mechanics}.
\newblock {\em ArXiv e-prints}, apr 2013.

\bibitem{landau1980statistical}
L.D. Landau, E.M. Lifshitz, and L.P. Pitaevskii.
\newblock {\em Statistical Physics: Theory of the Condensed State}.
\newblock Course of Theoretical Physics. Butterworth-Heinemann, 1980.

\bibitem{Lieb19991}
E.H. Lieb and J.~Yngvason.
\newblock The physics and mathematics of the second law of thermodynamics.
\newblock {\em Physics Reports}, 310(1):1 -- 96, 1999.

\bibitem{Shannon1948}
C.~E. Shannon.
\newblock A mathematical theory of communication.
\newblock {\em The Bell Systems Technical Journal}, 27:379--423, 623--656,
  1948.

\bibitem{story2002dynamics}
T.L. Story.
\newblock {\em Dynamics on Differential One-Forms}.
\newblock iUniverse, 2002.

\bibitem{1998kaisekirikigaku}
Y.~Yamamoto and K.~Nakamura.
\newblock {\em Analytical Mechanics 1}.
\newblock Asakura Publishing Co., Ltd., 1998.
\newblock (in Japanese).

\end{thebibliography}
\end{document}